\newif\ifarxiv\arxivtrue
\ifarxiv
\documentclass{amsart}
\else
\documentclass{kybernetika}
\fi

\usepackage[utf8]{inputenc}
\usepackage{amsmath}
\usepackage{amsthm}
\usepackage[sort, numbers]{natbib}
\usepackage{hyperref}
\usepackage{xcolor}
\usepackage{booktabs}
\usepackage{graphicx}
\usepackage[singlelinecheck=false,justification=justified]{caption}

\newcommand{\comment}[1]{}

\newtheorem{thm}{Theorem}[section]
\newtheorem{lemma}[thm]{Lemma}

\theoremstyle{definition}
\newtheorem{example}[thm]{Example} % [section]
\theoremstyle{remark}
\newtheorem{remark}[thm]{Remark} % [section]

\newcommand\Perp{\protect\mathpalette{\protect\independenT}{\perp}}
\newcommand{\ind}[3][]{\left.#2 \Perp\!\!_{\smash{#1}}\; #3\inD}
\newcommand{\inD}[1][\relax]{\def\argone{#1}\def\temprelax{\relax}%
  \ifx\argone\temprelax\right.\else\,\strut\middle|\,#1\right.{}\fi}

\def\independenT#1#2{\mathrel{\rlap{$#1#2$}\mkern2mu{#1#2}}}

\newcommand{\intDeltap}{\overset{\circ}{\Delta}_P}

\newcommand{\Rb}{\mathbb{R}}

\newcommand{\Tcal}{\mathcal{T}}

\newcommand{\Xcal}{\mathcal{X}}
\newcommand{\Ycal}{\mathcal{Y}}
\newcommand{\trans}{\top}
\DeclareMathOperator*{\bigtimes}{\textnormal{\Large $\times$}} % Cartesian Product
\DeclareMathOperator{\supp}{supp}

\DeclareMathOperator{\diag}{diag}

\title{Properties of Unique Information}

\ifarxiv
\author[Johannes Rauh]{Johannes Rauh$^{\dag}$}
\thanks{$^{\dag}$These authors contributed equally.}
\address{MPI für Mathematik in den Naturwissenschaften, Leipzig}
\email{jrauh@mis.mpg.de}
\author[Maik Schünemann]{Maik Schünemann$^{\dag}$}
\address{Institute for Theoretical Physics, University of Bremen}
\email{maik@neuro.uni-bremen.de}
\author{Jürgen Jost}
\address{MPI für Mathematik in den Naturwissenschaften, Leipzig}
\email{jjost@mis.mpg.de}
\else

\author{Johannes Rauh\thanks{These authors contributed equally.}, Maik Sch\"unemann\footnotemark[1] and J\"urgen Jost}

\contact{Johannes}{Rauh}{Max Planck Institute for Mathematics in the Sciences, Inselstr.~22, 04103 Leipzig Germany}{jrauh@mis.mpg.de}
\contact{Maik}{Sch\"unemann}{Department of Theoretical Physics, University of Bremen, Hochschulring 18, 28359 Bremen}{maik@neuro.uni-bremen.de}
\contact{J\"urgen}{Jost}{Max Planck Institute for Mathematics in the Sciences, Inselstr.~22, 04103 Leipzig Germany}{jjost@mis.mpg.de}
\fi
\date{\today}

\setcitestyle{square}

\begin{document}

\ifarxiv
\keywords{Information decomposition, unique information}
\subjclass[2010]{%
  94A15; % 94A15 - Information theory, general [See also 62B10, 81P45]
  94A17} % 94A17 - Measures of information, entropy
\else
\maketitle % for AMSART, maketitle must be after abstract
\fi

\begin{abstract}
  We study the unique information function $UI(T:X\setminus Y)$
  defined by \citet{bertschinger2014quantifying} within the framework
  of information decompositions. In particular, we study uniqueness
  and support of the solutions to the convex optimization problem
  underlying the definition of~$UI$. We identify sufficient conditions
  for non-uniqueness of solutions with full support in terms of
  conditional independence constraints and in terms of the cardinalities of
  $T$, $X$ and~$Y$. Our results are based on a reformulation of the
  first order conditions on the objective function as rank constraints
  on a matrix of conditional probabilities. These results help to
  speed up the computation of~$UI(T:X\setminus Y)$, most notably when
  $T$ is binary. Optima in the relative interior of the
  optimization domain are solutions of linear equations if $T$ is
  binary. In the all binary case, we obtain a complete picture of where
  the optimizing probability distributions lie.
\end{abstract}

\ifarxiv
\maketitle
\tableofcontents
\else
\keywords{Information decomposition, unique information}
\classification{94A15,94A17}
\fi

\section{Introduction}
\label{sec:intro}

\citet{bertschinger2014quantifying} introduced an information measure $UI(T:X\setminus Y)$ which they
called \emph{unique information}.  The function $UI$ is proposed within the framework of information
decompositions \citep{WilliamsBeer:Nonneg_Decomposition_of_Multiinformation} to quantify the amount
of information about $T$ that is contained in $X$ but not in~$Y$.  Similar quantities
within this framework have been proposed by
\citet{HarderSalgePolani2013:Bivariate_redundancy}, % $I_{\text{red}}$
\citet{Ince17:Iccs}, % $I_{\text{ccs}}$
\citet{JamesEmenheiserCrutchfield18:Idep} %$I_{\text{dep}}$ by
and \citet{NiuQuinn:IG_decomposition}.  Among them, the quantity $UI$ probably has the clearest
axiomatic characterization.  Although it has received a lot of attention by
theorists \citep[see e.g.][]{RBOJ14:Reconsidering_unique_information,BOJR18:UI_and_deficiencies,RBOJ19:UI_and_Secret_Key_Decompositions}, so far, applications have
focused on other measures, because $UI$ is difficult to compute, although there has been recent
progress \citep{MakkehTheisVicente17:PID_optimization_perspective,BanerjeeRauhMontufar18:Computing_the_UI}.

The function $UI$ is defined by means of an optimization problem.  Let $T$, $X$, $Y$ be random
variables with finite state spaces $\Tcal,\Xcal,\Ycal$ and with a joint distribution~$P$.  Let
$\Delta_{\Tcal,\Xcal,\Ycal}$ be the set of all joint distributions of such random variables, and let
\begin{multline*}
  \Delta_{P} = \Big\{
  Q\in\Delta_{\Tcal,\Xcal,\Ycal} : Q(X=x,T=t) = P(X=x,T=t), \\ Q(Y=y,T=t) = P(Y=y,T=t)
  \text{ for all }x\in\Xcal,y\in\Ycal,t\in\Tcal
  \Big\}
\end{multline*}
be the set of all joint distributions that have the same pair marginals as~$P$ for the pairs $(X,T)$
and~$(Y,T)$.  Then
\begin{equation}
  \label{eq:def_UI}
  UI(T:X\setminus Y) = \min_{Q\in\Delta_{P}}I_{Q}(T:X|Y),
\end{equation}
where $I_{Q}(T:X|Y)$ denotes the conditional mutual information of $T$ and $X$ given~$Y$, computed with
respect to~$Q$.
Due to the definition of~$\Delta_{P}$, the optimization problem in~\eqref{eq:def_UI} can be reformulated as follows:
\begin{equation}
  \label{eq:optimization}
  \min_{Q\in\Delta_{P}}I_{Q}(T:X|Y)
  = H(T|Y) - \max_{Q\in\Delta_{P}}H(T|X,Y).
\end{equation}

This paper studies $UI$, focusing on the following two questions:
\begin{enumerate}
\item When is there a unique solution to the optimization problems in~\eqref{eq:optimization}?
\item When is there a solution in the relative interior of~$\Delta_{P}$?
\end{enumerate}

In the framework of information decomposition, the solutions to the optimization
problems~\eqref{eq:optimization} are distributions with ``zero synergy about~$T$.''  Thus,
understanding these solutions sheds light on the concept of synergy.  If the solution is unique,
there is a unique way to combine the random variables $X$ and~$Y$ without synergy about~$T$ that
preserves the $(X,T)$- and $(Y,T)$-marginals.

Moreover, a unique solution $Q^{*}$ might be used to ``localize'' the
information decomposition, in the sense of~\citet{FinnLizier18:Pointwise_PI_using_SPAM}; although one
should keep in mind that the support of $Q^{*}$ may satisfy $\supp(Q^{*})\not\supseteq\supp(P)$.
A better understanding of the optimization problems also helps in the
computation of~$UI$. In the case where \(\Tcal\) is binary, an optimum
in the interior of \(\Delta_P\) can be found as solutions of linear
equations. Solving an optimization problem can be avoided in the all
binary case in which we derive a closed form solution of the
optimization problem.

\subsection*{Summary of results and outline}

Section~\ref{sec:structure-Delta_P} describes how the optimization domain~$\Delta_{P}$ and its support depend on \(P\).

Section~\ref{sec:support-uniqueness-optimum} summarizes general facts about the optimization problem.
The relationship between uniqueness of the optimizer and the supports of the optimizers is discussed, and sufficient conditions for non-uniqueness are identified.
% In particular, it is shown that under some conditions on the support of \(\Delta_P \), if both conditional independence statements \(\ind[P]{T}{X}\) and \(\ind[P]{T}{Y}\) hold, it follows that the optimum is not unique.

Section~\ref{sec:binary-T} specializes to the case where~$T$ is binary.  In this case, if there is an optimizer in the interior, then this optimizer satisfies a conditional independence constraint. In general, the optimizer is not unique.
We analyze how often the optimum lies in the interior or at the boundary of \(\Delta_P\) and how often an optimum in the interior is unique as a function of the cardinalities of \(\Xcal,\Ycal\) when sampling \(P\) uniformly from \(\Delta_{\Tcal,\Xcal,\Ycal} \).

Section~\ref{sec:binary-case} gives a complete picture for the case where all variables are binary.  In this case, $\Delta_{P}$ is a rectangle, a line segment or a single point.  A closed form expression is given for optimizers that lie in the interior of \(\Delta_P\).  If the optimizer does not lie in the interior, the optimum is attained at a vertex of~\(\Delta_P\).

Section~\ref{sec:examples} collects examples that demonstrate that the
conditions of some of our results are indeed necessary. The final
Section~\ref{sec:conclusions} presents our conclusions.

\section{The optimization domain \texorpdfstring{$\Delta_{P}$}{Delta-P}}
\label{sec:structure-Delta_P}

Fix a joint distribution~$P\in\Delta_{\Tcal,\Xcal,\Ycal}$.  Since the marginal of~$T$ is constant
on~$\Delta_{P}$, the support of~$T$, which we denote by $\Tcal' := \{t\in\Tcal : P(T=t)>0\}$, is
also constant on~$\Delta_{P}$.

% \subsubsection*{The optimization domain $\Delta_{P}$.}

Any distribution $Q\in\Delta_{P}$ is characterized uniquely by the conditional probabilities
$Q(X,Y|T=t)$ for~$t\in\Tcal'$.  The map
\begin{equation*}
  P\in\Delta_{\Tcal,\Xcal,\Ycal}\mapsto \big(P(X,Y|T=t)\big)_{t\in\Tcal'}
\end{equation*}
(where $\Tcal'$ depends on~$P$) induces a linear bijection
\begin{equation*}
  \Delta_{P} = % \bigcup_{\Tcal'\subseteq\Tcal,\Tcal'\neq\emptyset}
  \bigtimes_{t\in\Tcal'}\Delta_{P,t},
\end{equation*}
where
\begin{align*}
  \Delta_{P,t} = \Big\{ Q\in\Delta_{\Xcal,\Ycal} : Q(X=x) = P(X=x&|T=t),\\ Q(Y=y) = P(Y=y&|T=t) \Big\},
\end{align*}
and $\Delta_{\Xcal,\Ycal}$ is the set of all probability distributions of random variables \(X,Y \) with finite state spaces \(\Xcal,\Ycal\).
For example, when $X$ and $Y$ are binary, $\Delta_{P,t}$ is a line segment (which may degenerate to a point) for all~$t\in\Tcal'$.
Thus, $\Delta_{P}$ is a product of line segments; that is, a hypercube (up to a scaling).  If $T$ is also binary, then $\Delta_{P}$ is a
rectangle (a product of two line segments), which may degenerate to a
line segment or even a point depending on the support of $P$.
A figure of $\Delta_{P}$ in the case that all variables are binary (when $\Delta_{P}$ is a rectangle) can be found in\citep{bertschinger2014quantifying}.
Figure~\ref{fig:vis-223} makes use of the product structure to visualize $\Delta_{P}$ in the case $|\Tcal|=2=|\Xcal|$, $|\Ycal=3|$, where $\dim(\Delta_{P}) = 4$.

In the following, for $Q\in\Delta_{P}$ and $t\in\Tcal'$, we write $Q_{t}:=Q(X,Y|T=t)$ for the conditional distribution
of $X,Y$ given that~$T=t$.  The product structure of $\Delta_{P}$ implies: if $Q\in\Delta_{P}$ lies on the boundary of~$\Delta_{P}$,
then at least one of the $Q_{t}$ lies on the boundary of~$\Delta_{P,t}$.
Moreover, $Q$ lies on the boundary of~$\Delta_{\Tcal,\Xcal,\Ycal}$.  Hence,
the boundaries of the polytopes
$\Delta_{P}$ or $\Delta_{P,t}$ are characterized by the vanishing of probabilities.

\begin{remark}
  In the following, the expression \emph{boundary of $\Delta_{P}$} refers to the relative boundary.
  If $P$ lies on the boundary of~$\Delta_{\Tcal,\Xcal,\Ycal}$, then $\Delta_{P}$ may be a subset of
  the boundary of~$\Delta_{\Tcal,\Xcal,\Ycal}$ (cf.\ Lemma~\ref{lem:DeltaP_on_boundary}).  This happens if and
  only if one probability vanishes throughout $\Delta_{P,t}$ (and thus one probability vanishes
  throughout~$\Delta_{P}$).  In this case, $\Delta_{P}$ is part of the boundary
  of~$\Delta_{\Tcal,\Xcal,\Ycal}$.  However, the (relative) boundary of $\Delta_{P}$ is a strict
  subset of~$\Delta_{P}$, and the same holds for~$\Delta_{P,t}$.
\end{remark}

% \subsection{Coordinates in \texorpdfstring{$\Delta_{P}$}{DeltaP}.}

Let $A$ be the linear map that maps a joint distribution~$P\in\Delta_{\Tcal,\Xcal,\Ycal}$ to the pair $(P(X,T),P(Y,T))$ of marginal distributions.  Then
\begin{equation*}
  \Delta_{P} = (P + \ker(A)) \cap \Delta_{\Tcal,\Xcal,\Ycal}.
\end{equation*}

The difference of any two elements of $\Delta_{P}$ belongs to $\ker(A)$.  Conversely, the elements of $\ker(A)$  can be used to move within each~$\Delta_{P}$.
A generating set of $\ker(A)$ is given by the vectors
\begin{equation}\label{eq:def-gamma}
  \gamma_{t;x,x';y,y'} = \delta_{t,x,y} + \delta_{t,x',y'} - \delta_{t,x,y'} - \delta_{t,x',y},
  \qquad x,x'\in\Xcal, y,y'\in\Ycal,t\in \Tcal
\end{equation}
where $\delta_{t,x,y}$ denotes the Dirac measure supported at $T=t,X=x,Y=y$.
These vectors are not linearly independent.  One way to choose a linearly independent subset is to fix $x_{0}\in\Xcal$, $y_{0}\in\Ycal$.  Then the set
\begin{equation*}
  \Gamma := \big\{\gamma_{t;x,x_{0};y,y_{0}}:
  x\in\Xcal\setminus\{x_{0}\},
  y\in\Ycal\setminus\{y_{0}\},
  t\in \Tcal
  \big\}
\end{equation*}
is a basis of~$\ker(A)$.

\begin{remark}
  \label{rem:Markovbasis}
Apart from being symmetric, the larger dependent set has the following advantage, which is reminiscent of the Markov basis property \citep{DiaconisSturmfels98:Algebraic_Sampling}: Any
two points $Q,Q'\in\Delta_{P}$ can be connected by a path in $\Delta_{P}$ by applying a sequence of multiples of the
elements~$\gamma_{t;x,x';y,y'}$.  The same is not true if we restrict $x',y'$ to $x_{0},y_{0}$: if $Q(X=x_{0}) = 0$, then
adding a multiple of $\gamma_{t;x,x_{0};y,y_{0}}$ for any $x\in\Xcal$, $y\in\Ycal$ leads to a negative entry.
\end{remark}

Let $V$ be the set of distributions $Q_{0}\in\Delta_{\Tcal,\Xcal,\Ycal}$ that have a factorization of
the form
\begin{equation*}
  Q_{0}(t,x,y) = Q_{0}(t) Q_{0}(x|t) Q_{0}(y|t).
\end{equation*}
Thus, $V$ consists of all joint distributions that satisfy the Markov chain $X\;\text{--}\;T\;\text{--}\;Y$.
For each $P\in\Delta_{\Tcal,\Xcal,\Ycal}$, the intersection $\Delta_{P}\cap V$ contains precisely one element $Q_{0}=Q_{0}(P)$; namely
\begin{equation}
  \label{eq:Q0P}
  Q_{0}(t,x,y) = P(t) P(x|t) P(y|t).
\end{equation}
In the language of information geometry, $\Delta_{P}$ is a linear family that is dual to the exponential family~$V$~\citep{AmariNagaoka00:Methods_of_Information_Geometry}.
A general distribution \(Q \in \Delta_{\Tcal,\Xcal,\Ycal}\) can thus be expressed uniquely in the form
\begin{align}
\label{eq:decomposition}
  Q = Q_0 + \sum_{t,x',y'}P(t)g_{t,x',y'}\gamma_{t;x,x^0;y,y^0}
\end{align}
with \(Q_0=Q_{0}(Q) \in V\) and \(g =
(g_{t,x',y'})_{t,x'\neq x^0,y' \neq y^0}\)
denoting the coefficients with respect to \(\Gamma\). 

Let $\supp(\Delta_{P}) := \bigcup_{Q\in\Delta_{P}}\supp(Q)$ be the largest support of
an element of $\Delta_{P}$.  Generic elements of $\Delta_{P}$ have support $\supp(\Delta_{P})$.
We also let
\begin{multline*}
  \supp(\Delta_{P,t}) := \bigcup_{Q\in\Delta_{P}}\supp(Q_{t}) \\
  = \big\{ (x,y)\in\Xcal\times\Ycal : (t,x,y)\in\supp(\Delta_{P})\big\}\text{ for~$t\in\Tcal'$.}
\end{multline*}
If $\Delta_{P}$ is a singleton,
then $P=Q_{0}$.  In this case, $\supp(\Delta_{P}) = \supp(P)$, and $\supp(\Delta_{P,t}) = \supp(P_{t})$.

For $t\in\Tcal'$ let $\Xcal_{t}=\big\{x\in\Xcal: P(X=x|T=t)>0\big\}$ and
$\Ycal_{t}= \big\{y\in\Ycal: P(Y=y|T=t)>0\}$.
It follows from the definitions:
\begin{lemma}
  \label{lem:suppDeltaPt}
  Let $t\in\Tcal'$.  Then
  $\supp(\Delta_{P,t}) = \supp(Q_{0,t}) = \Xcal_{t}\times\Ycal_{t}$.
  Moreover, $\supp(\Delta_{P}) = \supp(Q_{0})$.  Thus, $Q_{0}$ has maximal support in~$\Delta_{P}$.
\end{lemma}
% \begin{proof}
%   The second equality follows from~\eqref{eq:Q0P}.
%   The containment
%   \begin{equation*}
%     \supp(\Delta_{P,t})
%     \subseteq \big\{x\in\Xcal: P_{t}(X=x)>0\big\} \times \big\{y\in\Ycal: P_{t}(Y=y)>0\}
% %    = \big\{(x,y) : \text{ there exist }x'\in\Xcal,y'\in\Ycal\text{ with }(x,y'),(x',y)\in\supp(P_t) \big\}.
%   \end{equation*}
%   follows since any $Q\in\supp(\Delta_{P,t})$ has the same $X$-marginal and the same $Y$-marginal
%   as~$P_{t}$.  For the converse containment, let
%   \begin{equation*}
%     Q(x,y) = P_{t}(X=x)\cdot P_{t}(Y=y)\text{ for all }x\in\Xcal, Y\in\Ycal.
%   \end{equation*}
%   Then $Q\in\Delta_{P,t}$, and
%   \begin{equation*}
%     \supp(Q)
%     = \big\{x\in\Xcal: P(X=x)>0\big\} \times \big\{y\in\Ycal: P(Y=y)>0\}.\qedhere
%   \end{equation*}
% \end{proof}
The next lemma follows from Lemma~\ref{lem:suppDeltaPt} and the definitions:
\begin{lemma}\label{lem:DeltaP_on_boundary}
  Let $t\in\Tcal$, $x\in\Xcal$ and $y\in\Ycal$.
  The following statements are equivalent:
  \begin{enumerate}
  \item $\Delta_{P}$ lies in the face of $\Delta_{\Tcal,\Xcal,\Ycal}$ defined by $Q(t,x,y) = 0$.
  \item $(t,x,y)\notin \supp(\Delta_{P})$.
  \item Every $Q\in\Delta_{P}$ satisfies $Q(t,x,y)=0$.
  \item $Q_{0}:=Q_{0}(P)$ satisfies $Q_{0}(t,x,y)=0$.
  % \item One of the two sets $\{x'\in\Xcal : P(t,x',y)>0\}$ and $\{y'\in\Xcal :
  % P(t,x,y')>0\}$ is empty (or both).
  \item $P(T=t,Y=y)P(T=t,X=x)=0$.
  \end{enumerate}
\end{lemma}
%%% JR: I commented out the proof.  Most of the proof is trivial.  After adding some statements, the numbering got mixed up, and I'm not motivated to maintain the proof now.
% \begin{proof}
%   The equivalence of (1) and (2) is by definition.
%   The equivalence of (1) and (3) follows from the discussion above the Lemma.
%   To see that (3) is equivalent, observe that the assumption implies that for any
%   $x'\in\Xcal\setminus\{x\}$, $y'\in\Ycal\setminus\{y\}$, either $Q^{*}(t,x,y')=0$ or $Q^{*}(t,x',y)=0$; for otherwise,
%   $Q_{\epsilon}:=Q^{*}+\epsilon\gamma_{t;x,x';y,y'}$ for small $\epsilon>0$ would be an element in~$\Delta_{P}$ with
%   $Q_{\epsilon}(t,x,y)=\epsilon>0$.
% \end{proof}

\begin{lemma}
  \label{lem:Delta_Ps_singleton}
  Let $t\in\Tcal'$.  The following are equivalent:
  \begin{enumerate}
  \item $\Delta_{P,t}$ is a singleton.
  \item At least one of $\Xcal_{t}$, $\Ycal_{t}$ is a singleton.
  % \item Either
  %   $\{ x\in\Xcal : xy\in\supp(P_{t}) \text{ for some }y\in\Ycal\}$
  %    or 
  %   $\{ y\in\Ycal : xy\in\supp(P_{t}) \text{ for some }x\in\Xcal\}$
  %   is a singleton.
  % \item Either
  %   $\{ x\in\Xcal : xy\in\supp(\Delta_{P,t}) \text{ for some }y\in\Ycal\}$
  %    or
  %   $\{ y\in\Ycal : xy\in\supp(\Delta_{P,t}) \text{ for some }x\in\Xcal\}$
  %   is a singleton.
  \end{enumerate}
\end{lemma}
\begin{proof}
  Condition~2.\ in the lemma captures precisely when it is not possible to add a multiple of some $\gamma_{t;x,x';y,y'}$
  to~$P$ or, in fact, to any~$Q\in\Delta_{P}$ (cf. Remark~\ref{rem:Markovbasis}).
\end{proof}

\section{Support and uniqueness of the optimum}
\label{sec:support-uniqueness-optimum}

This section studies the uniqueness of the optimizer and the question, when it lies on the boundary of~$\Delta_{P}$.
There are many relations between uniqueness and support of the optimizers:
Lemma~\ref{lem:uniqueness_and_support} states that, if the optimizer is not unique, then there are
optimizers with restricted support.  Theorems~\ref{thm:non-uniqueness-cardinalities},
\ref{thm:non-uniqueness-cardinalities-CI}, \ref{thm:non-uniqueness}
and~\ref{thm:non-uniqueness-v2} prove that either the optimizer lies at the boundary or it is not
unique under a variety of different assumptions that involve the cardinalities of $|\Xcal|$,
$|\Ycal|$ and $|\Tcal|$ or conditional independence conditions.

\begin{lemma}
  \label{lem:uniqueness_and_support}
  If the optimizer is not unique, then there exists an optimizer on the boundary of $\Delta_{P}$.
\end{lemma}
\begin{proof}
  Suppose that there are two distinct optimizers $Q_{1},Q_{2}\in\Delta_{P}$, and assume that neither
  $Q_{1}$ nor $Q_{2}$ lies on the boundary of $\Delta_{P}$.  By convexity of the target function
  $I_{Q}(T:X|Y)$ on $\Delta_{P}$ (see Lemma~4 in \citep{bertschinger2014quantifying} or
  Lemma~\ref{lem:H-concave} below), the convex hull of $Q_{1}$ and $Q_{2}$ consists of optimizers.
  Let $L_{Q_{1},Q_{2}}$ be the line through $Q_{1},Q_{2}$.  The target function $I_{Q}(T:X|Y)$ is a
  continuous function on the line segment $L_{Q_{1},Q_{2}}\cap\Delta_{P}$, and it is analytic on the
  relative interior of this line segment.  By assumption, $I_{Q}(T:X|Y)$ is constant on the part of
  $L_{Q_{1},Q_{2}}$ between $Q_{1}$ and~$Q_{2}$.  By the principle of permanence, $I_{Q}(T:X|Y)$ is
  constant on $L_{Q_{1},Q_{2}}\cap\Delta_{P}$.  Therefore, the two points where $L_{Q_{1},Q_{2}}$
  intersect the boundary of $\Delta_{P}$ are optimizers of $I_{Q}(T:X|Y)$ that lie on the boundary
  of~$\Delta_{P}$.
\end{proof}

The derivative of $I_{Q}(T:X|Y)$ in the direction of $\gamma_{t;x,x';y,y'}$ at $Q$ equals
\begin{equation}
  \label{eq:pd}
  \log\left(\frac{Q(t,x,y)Q(t,x',y')}{Q(t,x,y')Q(t,x',y)}\cdot\frac{Q(x,y')Q(x',y)}{Q(x,y)Q(x',y')}\right)
  = \log\left(\frac{Q(t|x,y)Q(t|x',y')}{Q(t|x,y')Q(t|x',y)}\right),
\end{equation}
assuming that the probabilities in the logarithm are positive.  Otherwise, the partial derivative has to be computed as
a limit.

\begin{remark}
  \label{rem:sum-of-rk-one}
  The vanishing of the directional derivative of $I_{Q}(T:X|Y)$ can be seen as a determinantal
  condition: all derivatives~\eqref{eq:pd} vanish if and only if for all $t\in\Tcal'$ the determinants
  of all $2\times2$-submatrices of the matrix $(Q(t|x,y))_{x,y}\in\Rb^{\Xcal\times\Ycal}$ vanish;
  that is, if and only if these matrices have rank one.
  As $\sum_{t\in\Tcal'}Q(t|x,y)=1$ for all~$x,y$, the sum of these rank-one matrices is again of
  rank one.

  Conversely, let $\tilde Q_{1},\dots,\tilde Q_{k}$ be non-negative rank-one matrices such that the
  sum $\tilde Q = \tilde Q_{1}+\dots+\tilde Q_{k}$ is non-zero and again of rank one; say $\tilde Q=vw^{\trans}$
  with $v,w$ non-negative and \(\trans\) denoting the transpose.  Let $V=\diag(v)$, $W=\diag(w)$, and let $Q_{t}=V^{-1}\tilde Q_{t}W^{-1}$ for
  $t=1,\dots,k$.  Then $Q_{1}+\dots+Q_{k}=V^{-1}\tilde QW^{-1}$ is the matrix with all entries equal
  to one.  Thus, the matrices $Q_{t}$ for $t=1,\dots,k$ can be interpreted as matrices of
  conditional probabilities~$Q(t|X,Y)$.  Together with any distribution of the pair $(X,Y)$, one
  obtains a distribution $Q(T,X,Y)$ at which all directional derivatives of $I_{Q}(T:X|Y)$ vanish.
\end{remark}
\begin{lemma}
  \label{lem:maxi-boundary}
  Let $Q^{*}$ be a minimizer of $I_{Q}(T:X|Y)$ for $Q\in\Delta_{P}$, and
  let $(t,x,y)\in\supp(\Delta_{P})$.
  If $Q^{*}(t,x,y) = 0$, then $Q^{*}(x,y) = 0$.  Thus, $Q^{*}(t',x,y) = 0$ for all $t'\in\Tcal$.
\end{lemma}
\begin{proof}
  Suppose that $Q^{*}(t,x,y) = 0$, but that $Q^{*}(x,y)>0$.  Then there exist $x',y'$ such that $Q_{\epsilon} := Q^{*} +
  \epsilon\gamma_{t;x,x';y,y'}$ is non-negative for $\epsilon>0$ small enough (and thus $Q_{\epsilon}\in\Delta_{P}$).  In particular, $Q^{*}(t,x',y), Q^{*}(t,x,y')>0$.

  Since $Q^{*}$ is a minimizer, the partial derivative~\eqref{eq:pd} at~$Q^{*}$ must be non-negative.  Note that, by
  assumption, $Q^{*}(t,x,y) = 0$.  If all four probabilities in the denominator of the fraction in the logarithm were
  non-zero, then the partial derivative would be equal to minus infinity.  Thus, either $Q^{*}(x,y)$ or $Q^{*}(x',y')$ must
  vanish.

  Suppose that $Q^{*}(x,y)>0$.  Then $Q^{*}(x',y')=0$.  Hence, $Q^{*}(t,x',y')=0$, and so
  \begin{multline*}
    \frac{Q_{\epsilon}(t,x,y)Q_{\epsilon}(t,x',y')}{Q_{\epsilon}(t,x,y')Q_{\epsilon}(t,x',y)}\cdot\frac{Q_{\epsilon}(x,y')Q\epsilon(x',y)}{Q_{\epsilon}(x,y)Q_{\epsilon}(x',y')}
    \\
    = \frac{\epsilon^{2}Q_{\epsilon}(x,y')Q_{\epsilon}(x',y)}{Q_{\epsilon}(t,x,y')Q_{\epsilon}(t,x',y)Q_{\epsilon}(x,y)\epsilon}
    = O(\epsilon).
  \end{multline*}
  Thus, the partial derivative diverges as $\log(\epsilon)$ to $-\infty$ as $\epsilon\to0$, contradicting the fact that
  $Q^{*}$ is a local minimizer.  Therefore, $Q^{*}(x,y)=0$.
\end{proof}

If $Q^{*}(t,x,y)=0$ and $Q^{*}(t,x',y)>0$, $Q^{*}(t,x,y')>0$ for some $t\in\Tcal'$, $x,x'\in\Xcal$, $y,y'\in\Ycal$, then the partial derivative at $Q^{*}$ in the direction of $\gamma_{t;x,x';y,y'}$ is
\begin{equation*}
  \log\left(
    \frac{Q^{*}(t,x',y')Q^*(x,y')Q^*(x',y)}{Q^*(t,x,y')Q^*(t,x',y) Q^*(x',y')}
  \right).
\end{equation*}
Therefore,
\begin{equation*}
  Q^{*}(t,x',y')Q^*(x,y')Q^*(x',y) \ge Q^*(t,x,y')Q^*(t,x',y) Q^*(x',y'),
\end{equation*}
or
\begin{equation*}
  \frac{Q^{*}(t,x',y')}{Q^*(x',y')} \ge \frac{Q^*(t,x,y')}{Q^{*}(x,y')}\frac{Q^*(t,x',y)}{Q^{*}(x',y)}.
\end{equation*}

% \section{Uniqueness}
% \label{sec:uniqueness}

It is well known that entropy is strictly concave and that conditional entropy is concave.  From the proof of this fact,
it is easy to analyze where conditional entropy is strictly concave.
\begin{lemma}
  \label{lem:H-concave}
  The conditional entropy $H(A|B)$ is concave in the joint distribution of~$A,B$.  It is strictly concave, with the
  exception of those directions where $P(A|B)$ is constant.  That is:
  \begin{equation*}
    \lambda H_{P_{1}}(A|B) + (1-\lambda) H_{P_{2}}(A|B) \le H_{\lambda P_{1} + (1-\lambda) P_{2}}(A|B)
  \end{equation*}
  with equality if and only if~$P_{1}(A|B) = P_{2}(A|B)$ a.e.
\end{lemma}
\begin{proof}
  Let $\theta$ be a Bernoulli random variable with parameter~$\lambda$, and consider the joint distribution $P$ of
  $\theta$, $A$ and~$B$ given by
  \begin{equation*}
    P(A,B,\theta) = \begin{cases}
      \lambda P_{1}(A,B), & \text{ if }\theta = 0, \\
      (1 - \lambda) P_{2}(A,B), & \text{ if }\theta = 1.
    \end{cases}
  \end{equation*}
  Then
  \begin{multline*}
    H_{\lambda P_{1} + (1-\lambda) P_{2}}(A|B)
    = H_{P}(A|B)
    \ge H_{P}(A|B,\theta) \\
    = \lambda H_{P_{1}}(A|B) + (1-\lambda) H_{P_{2}}(A|B).
  \end{multline*}
  Equality holds if and only if $A$ is independent of $\theta$ given~$B$; that is:
  \begin{equation*}
    P_{1}(A|B) = P(A|B,\theta = 0) = P(A|B,\theta = 1) = P_{2}(A|B). \qedhere
  \end{equation*}
\end{proof}
\begin{lemma}
  \label{lem:unique-conditional}
  Let $Q_{1},Q_{2}\in\Delta_{P}$ be two maximizers of $\max_{Q\in\Delta_{P}}H_{Q}(T|XY)$.  Then $Q_{1}(T|XY) = Q_{2}(T|XY)$.
\end{lemma}
\begin{proof}
  We may assume that $Q_{1}\neq Q_{2}$.
  By assumption, $H_{Q}(T|XY)$ is constant on the line segment between $Q_{1}$ and~$Q_{2}$.
  Thus, on this line segment $H_{Q}(T|XY)$ is not strictly concave.
  By Lemma~\ref{lem:H-concave}, $Q_{1}(T|XY) = Q_{2}(T|XY)$.
\end{proof}

The following four theorems give different sufficient conditions for non-uniqueness of the optimizer.

\begin{thm}
  \label{thm:non-uniqueness-cardinalities}
  Suppose that $|\Tcal|<\min\big\{|\Xcal|,|\Ycal|\big\}$.  If there exists an optimizer of
  $\max_{Q\in\Delta_{P}} H_{Q}(T|XY)$ with full support, then the optimizer is not unique.
\end{thm}
\begin{proof}
  Suppose that $Q^{*}\in\arg\max_{Q\in\Delta_{P}} H_{Q}(T|XY)$ has full support.
  The proof proceeds by finding a direction within $\Delta_{P}$ in which $H_{Q}(T|XY)$ is not
  strictly concave.
  Consider the linear equation
  \begin{equation}
    \label{eq:cond-dist}
    Q(t,x,y) = Q^{*}(t|x,y) Q(x,y)
    \quad
    \text{ for }Q\in\Delta_{P}.
  \end{equation}
  If $Q'\in\Delta_{P}$ solves this equation, then, by Lemma~\ref{lem:H-concave}, the function
  $H_{Q}(T|X,Y)$ is affine on the line connecting $Q^{*}$ and $Q'$.  Since $Q^{*}$ is a maximizer,
  $H_{Q}(T|X,Y)$ is constant on this line, whence any point on this line is a maximizer.  Thus, to
  prove the theorem, it suffices to show that there exists a solution~$Q'\neq Q*$ in~$\Delta_{P}$
  to~\eqref{eq:cond-dist}.

  By Remark~\ref{rem:sum-of-rk-one}, for every $t\in\Tcal'$, there exists a pair of non-negative
  vectors $v_{t},w_{t}$ such that  $Q^{*}(t|x,y) = v_{t}w_{t}^{\trans}$.  The assumption
  $|\Tcal|<\min\big\{|\Xcal|,|\Ycal|\big\}$ implies that there exist non-zero $v_{0}\in\Rb^{\Xcal}$,
  $w_{0}\in\Rb^{\Ycal}$ with $v_{0}^{\trans}v_{t}=0=w_{0}^{\trans}w_{t}$ for all~$t\in\Tcal'$.  For
  $\epsilon\in\Rb$ let
  \begin{equation*}
    Q_{\epsilon}(x,y) := Q^{*}(x,y) + \epsilon v_{0,x}^{\trans}w_{0,y}.
  \end{equation*}
  Then
  \begin{equation*}
    \sum_{x\in\Xcal,y\in\Ycal}Q_{\epsilon}(x,y) =
    \sum_{x\in\Xcal,y\in\Ycal}Q^{*}(x,y) + \epsilon\sum_{x\in\Xcal,y\in\Ycal}v_{0,x}w_{0,y} = 1,
  \end{equation*}
  because
  \begin{multline*}
    \sum_{x\in\Xcal,y\in\Ycal}v_{0,x}w_{0,y}
    = \sum_{x\in\Xcal,y\in\Ycal}v_{0,x}w_{0,y}\sum_{t\in\Tcal'}Q^{*}(t|x,y)
    \\
    = \sum_{t\in\Tcal'} \sum_{x\in\Xcal}v_{0,x} v_{t,x}\sum_{y\in\Ycal}w_{0,y}w_{t,y} = 0.
  \end{multline*}
  Therefore, if $\epsilon$ is sufficiently close to zero, then $Q_{\epsilon}$ defines a probability
  distribution for $X$ and~$Y$.

  Extend $Q_{\epsilon}$ to a joint distribution of $T,X,Y$ by
  $Q_{\epsilon}(t,x,y) = Q^{*}(t|x,y) Q_{\epsilon}(x,y)$.  Then $Q_{\epsilon}$
  satisfies~\eqref{eq:cond-dist}.  It remains to show that $Q_{\epsilon}\in\Delta_{Q^{*}}$.
  From
  \begin{multline*}
    Q_{\epsilon}(t,x) - Q^{*}(t,x) = \sum_{y\in\Ycal}\big(Q_{\epsilon}(t,x,y) - Q^{*}(t,x,y)\big)
    \\
    = \sum_{y\in\Ycal} Q^{*}(t|x,y) \big( Q_{\epsilon}(x,y) - Q^{*}(x,y) \big)
    = \epsilon v_{t,x}v_{0,x} \sum_{y\in\Ycal}w_{t,y}w_{0,y} = 0
  \end{multline*}
  follows $Q_{\epsilon}(T,X)=Q^{*}(T,X)$.  The equality $Q_{\epsilon}(T,Y)=Q^{*}(T,Y)$ follows
  similarly.
\end{proof}
\begin{thm}
  \label{thm:non-uniqueness-cardinalities-CI}
  Let $|\Tcal|<|\Ycal|$, and suppose that $UI(T:X\setminus Y) = 0$.  If there is an optimizer of
  $\max_{Q\in\Delta_{P}} H_{Q}(T|X,Y)$ with full support, then the optimizer is not unique.
\end{thm}
\begin{proof}
  The proof of Theorem~\ref{thm:non-uniqueness-cardinalities} can be adapted.  Under the assumptions
  of the theorem, if $Q^{*}$ is an optimizer, then $Q^{*}(t|x,y) = Q^{*}(t|y)$ does not depend
  on~$x$.  Therefore, one may choose $v_{t,x}=1$ for all~$y\in\Ycal,t\in\Tcal$ and
  $w_{t,y}=Q^{*}(t|y)$.  To construct $v_{0}$, it now suffices that $|\Xcal|\ge 2$, since all
  vectors $v_{t}$, $t\in\Tcal$, are identical.
\end{proof}

\begin{thm}
  \label{thm:non-uniqueness}
  Suppose that $H(X),H(Y) > 0$.  If both \(\ind[P]{T}{X}\) and \(\ind[P]{T}{Y}\),
  then \(\arg\max_{Q\in \Delta_P}{H_Q(T|X,Y)}\) is not unique.
\end{thm}

\begin{proof}
Let \(Q_0 = Q_{0}(P) = P_TP_{X|T}P_{Y|T} =
P_TP_{X}P_{Y}\in\Delta_{P}\).  Then \(\ind[Q_0]{T}{(X,Y)}\) by construction.
Since $H_{Q}(T|X,Y) \le H(T)$ for $Q\in\Delta_{P}$ and since
$Q_{0}$ achieves equality, $Q_0$ maximizes
\(H(T|X,Y)\) on~$\Delta_{P}$.

Due to the assumption of positive entropy, there exist $x_{0},x_{1}\in\Xcal$, $y_{0},y_{1}\in\Ycal$ with $P_{X}(x_{0})>0$, $P_{X}(x_{1})>0$, $P_{Y}(y_{0}) > 0$ and $P_{Y}(y_{1}) > 0$.
For $\delta\in\Rb$ let
\begin{equation*}
  Q_{\delta}(t,x,y) := Q_{0}(t,x,y) + \delta p_{T}(t) \gamma_{t;x_{0},x_{1};y_{0},y_{1}}.
\end{equation*}
If $|\delta|$ is small enough, then $Q_{\delta}$ is non-negative and hence belongs to $\Delta_{P}$.  For such~$\delta$, the conditional $Q_{\delta}(x,y|t)$ does not depend on~$t$, whence $\ind[Q_{\delta}]{T}{(X,Y)}$.
Thus, all such $Q_{\delta}$ are maximizers of $H_{Q}(T|X,Y)$ for $Q\in\Delta_{P}$.
\end{proof}

\begin{example}
  \label{ex:binary-non-uniqueness}
  Let $P$ be the distribution of three independent uniform binary random variables $T,X,Y$, and let $P'$ be the joint distribution where $X,T$ are uniform independent binary random variables and where $X=Y$.  Then $\Delta_{P}=\Delta_{P'}$, and both $P$ and $P'$ maximize $H_{Q}(T|X,Y)$ for $Q\in\Delta_{P}$.

  This example is the same as Example~31 by~\citet{bertschinger2014quantifying}.  Ironically, \citet{bertschinger2014quantifying} remarked that the optimization problem is ill-conditioned, but they failed to observe the non-uniqueness of the optimum in this case.
\end{example}

The following technical result generalizes Theorem~\ref{thm:non-uniqueness}.  It is illustrated by Example~\ref{ex:blockwise-non-uniqueness}.
\begin{thm}
  \label{thm:non-uniqueness-v2}
  Suppose that \(\ind[P]{T}{X}[Y]\) and \(\ind[P]{T}{Y}[X]\).
  If there exist $x_{0}\in\Xcal$, $y_{0}\in\Ycal$ with $P(X=x_{0},Y=y_{0})>0$ and
  $H(X|Y=y_{0})\neq 0\neq H(Y|X=x_{0})$, then \(\max_{Q\in\Delta_P}{H_Q(T|X,Y)}\) is not unique.
\end{thm}
\begin{proof}
  If \(\ind[P]{T}{X}[Y]\), then $I_{P}(T:X|Y)=0$.  From this it follows that
  $P$ belongs to $\arg\min_{Q\in\Delta_{P}}I_{Q}(T:X|Y)=\arg\max_{Q\in\Delta_{P}}H_{Q}(T|XY)$.  The probability distributions that satisfy
  \(\ind[P]{T}{X}[Y]\) and \(\ind[P]{T}{Y}[X]\) have first been characterized
  by~\citet{Fink11:Binomial_ideal_of_intersection_axiom}; see also the reformulation
  by~\citet{RauhAy14:Robustness_and_systems_design}.  This characterization implies that there
  are partitions $\Xcal=\Xcal'_{1}\cup\dots\cup\Xcal'_{b}$ and
  $\Ycal=\Ycal'_{1}\cup\dots\cup\Ycal'_{b}$ such that
  $\supp(P)\subseteq(\Xcal'_{1}\times\Ycal'_{1})\cup\dots\cup(\Xcal'_{b}\times\Ycal'_{b})$ and such that
  $\ind[P]T{\{X,Y\}}[X\in\Xcal'_{i},Y\in\Ycal'_{i}]$ for $i=1,\dots,b$.  There exists
  $i_{0}\in\{1,\dots,b\}$ such that $x_{0}\in\Xcal'_{i_{0}}$ and $y_{0}\in\Ycal'_{i_{0}}$.  Since
  $H(X|Y=y_{0})\neq 0\neq H(Y|X=x_{0})$, there exist $x_{1}\in\Xcal'_{i_{0}}\setminus\{x_{0}\}$ and
  $y_{1}\in\Ycal'_{i_{0}}\setminus\{y_{0}\}$ with $P(x_{1},y_{0})>0$ and $P(x_{0},y_{1})>0$.  For
  $\delta>0$ let
  \begin{equation*}
    P_{\delta} = P + \delta\cdot P(T|X,Y)\gamma_{t;x_{0},x_{1};y_{0};y_{1}}.
  \end{equation*}
  If $\delta$ is positive and small enough, then $P_{\delta}$ is a probability distribution in
  $\Delta_{P}$ that satisfies $\supp(P)=\supp(P_{\delta})$.  Moreover,
  $\ind[P_{\delta}]T{\{X,Y\}}[X\in\Xcal'_{i},Y\in\Ycal'_{i}]$ for $i=1,\dots,b$.  Hence,
  \(\ind[P_{\delta}]{T}{X}[Y]\) and \(\ind[P_{\delta}]{T}{Y}[X]\), and so
  $P_{\delta}\in\arg\min_{Q\in\Delta_{P}}I_{Q}(T:X|Y)$.
\end{proof}

\section{The case of binary \texorpdfstring{$T$}{T}}
\label{sec:binary-T}

\subsection{Independence properties for optimizers in the interior}
\label{sec:binary-T-independence}

If $\ind[P]{T}{X}[Y]$ or $\ind[P]{T}{Y}[X]$, then $P$ solves the PID optimization
problem~\eqref{eq:optimization}.  The next theorem is a partial converse in the case of binary~$T$.
We denote the interior of \(\Delta_{P} \) by \(\intDeltap\).

\begin{thm}
  \label{thm:cond_indep_binary_T}
  Let $T$ be binary. Assume that $\Delta_P$ has full support and that \(\tilde{Q} \in
  \intDeltap \cap \arg\max_{Q\in \Delta_{P}}{H_Q(T|X,Y)}\) is an interior point.
  Then, either \(\ind[\tilde{Q}]{T}{X}[Y]\) or \(\ind[\tilde{Q}]{T}{Y}[X]\) (or both).
  Thus, either $UI(T:X\setminus Y) = 0$ or $UI(T:Y\setminus X) = 0$.
\end{thm}
\begin{remark}
  The proof of the theorem relies on the vanishing condition of the directional derivatives. Thus,
  the conclusion still holds when $\tilde Q$ does not belong to $\intDeltap$,
  as long as all directional derivatives of the target
  function $H_{Q}(T|X,Y)$ exist and vanish at~$\tilde Q$.
  By Remark~\ref{rem:sum-of-rk-one}, this happens if and only if for any $t\in\Tcal'$ the matrix
  $(\tilde Q(t|x,y))_{x,y}\in\Rb^{\Xcal\times\Ycal}$ has rank one.
\end{remark}
\begin{remark}
  When $\Tcal$ has cardinality three or more, the statement of the theorem becomes false; see Example~\ref{ex:ternary-no-CI}.  This is related to the fact that there exist three positive rank-one-matrices the sum of which has again rank one, cf.\ Remark~\ref{rem:sum-of-rk-one}.
  When the support of $\Delta_{P}$ is not full, the statement of the theorem becomes false, even when all variables are binary; see Example~\ref{ex:binary-singleton}
\end{remark}
\begin{remark}
  Theorem \ref{thm:cond_indep_binary_T} can be used to efficiently
  compute $UI$ (and the corresponding bivariate information
  decomposition) when the optimum lies in the interior of~\(\Delta_P \),
  as searching for conditional independences in
  \(\Delta_P \) constitutes solving a linear programming problem (see
  the proof of Theorem~\ref{thm:inner-uniqueness}). If no solution in the interior is found,
  \(\max_{Q\in \partial \Delta_P}(H_Q(T|X,Y))\) has to be solved.
\end{remark}

\begin{proof}

Under the assumption that the optimum is attained in the interior of
\(\Delta_{P}\), it is characterized by \(\frac{\partial
   H_q(T|X,Y)}{\partial_{g_{t,x,y}}} = 0\).
 This leads  to the system of equations
\begin{align*}
 \log \frac{\tilde{Q}(t|x,y_{0}) \tilde{Q}(t|x_{0},y)}{\tilde{Q}(t|x_{0},y_{0}) \tilde{Q}(t|x,y)} &= 0 \text{ ,}
\end{align*}
for $t\in\{0,1\}$, $x\in\Xcal\setminus\{x_{0}\}$ and $y\in\Ycal\setminus\{y_{0}\}$.
For fixed $x,y$, this rewrites to
\begin{align*}
\tilde{Q}(0|x,y_{0}) \tilde{Q}(0|x_{0},y) &= \tilde{Q}(0|x_{0},y_{0}) \tilde{Q}(0|x,y) \\ 
\tilde{Q}(1|x,y_{0}) \tilde{Q}(1|x_{0},y) &= \tilde{Q}(1|x_{0},y_{0}) \tilde{Q}(1|x,y) \text{ .}
\end{align*}
Using \(\tilde{Q}(0|x,y) = 1-\tilde{Q}(1|x,y) \), this system is equivalent  to 
\begin{align*}
\tilde{Q}(0|x,y_{0}) \tilde{Q}(0|x_{0},y) &= \tilde{Q}(0|x_{0},y_{0}) \tilde{Q}(0|x,y) \\ 
\tilde{Q}(0|x,y_{0}) +\tilde{Q}(0|x_{0},y) &= \tilde{Q}(0|x_{0},y_{0}) +\tilde{Q}(0|x,y) \text{ .}
\end{align*}
These equations imply
\begin{align*}
  (\tilde{Q}(0|x,y_{0}) &- \tilde{Q}(0|x_{0},y_{0}))(\tilde{Q}(0|x_{0},y) - \tilde{Q}(0|x_{0},y_{0}))
  \\
                        &= \tilde{Q}(0|x,y_{0})\tilde{Q}(0|x_{0},y) - \tilde{Q}(0|x,y_{0})\tilde{Q}(0|x_{0},y_{0})
  \\ & \qquad\qquad - \tilde{Q}(0|x_{0},y_{0}))\tilde{Q}(0|x_{0},y) + \tilde{Q}(0|x_{0},y_{0}))^{2}
  \\ & = \tilde{Q}(0|x_{0},y_{0}) \big(
  \tilde{Q}(0|x,y) - \tilde{Q}(0|x,y_{0}) - \tilde{Q}(0|x_{0},y) + \tilde{Q}(0|x_{0},y_{0})
  \big) = 0.
\end{align*}
Therefore, for fixed values of $x$ and~$y$, there are only two possible solutions:
\begin{align*}
I(x,y): \tilde{Q}(t|x_{0},y_{0}) &= \tilde{Q}(t|x,y_{0}) \text{ and } \tilde{Q}(t|x,y) = \tilde{Q}(t|x_{0},y) \text{ for all }t,\\
I\!I(x,y): \tilde{Q}(t|x_{0},y_{0}) &= \tilde{Q}(t|x_{0},y) \text{ and } \tilde{Q}(t|x,y) = \tilde{Q}(t|x,y_{0}) \text{ for all $t$.}
\end{align*}

Let $\Xcal'=\Xcal\setminus\{x_{0}\}$ and $\Ycal'=\Ycal\setminus\{y_{0}\}$.  By what has been shown so far,
$A_{I}\cup A_{I\!I}=\Xcal'\times\Ycal'$, where
\begin{align*}
A_{I} &= \big\{(x,y) \in\Xcal'\times\Ycal' : I(x,y)\text{ holds}\big\}, \\
A_{I\!I} &= \big\{(x,y) \in\Xcal'\times\Ycal' : I\!I(x,y)\text{ holds}\big\}.
\end{align*}
We next show that either $A_{I}=\Xcal'\times\Ycal'$ or $A_{I\!I}=\Xcal'\times\Ycal'$ (or both).

Suppose that $A_{I}$ is not empty.  Let $(x,y)\in A_{I}$, and let $y'\in\Ycal'\setminus\{y\}$.  If
$I\!I(x,y')$ holds, then
$\tilde Q(t|x,y') = \tilde Q(t|x,y_{0}) = \tilde Q(t|x_{0},y_{0}) = \tilde Q(t|x_{0},y')$.  Thus,
$I(x,y')$ also holds, which implies $(x,y')\in A_{I}$.  Thus, $A_{I}\subset\Xcal'\times\Ycal'$ is of
the form $A_{I}=\Xcal'_{I}\times\Ycal'$, where $\Xcal'_{I}\subseteq\Xcal'$.

Similarly, $A_{I\!I} = \Xcal'\times\Ycal'_{I\!I}$, where $\Ycal'_{I\!I}\subseteq\Ycal'$.  If
$A_{I}\neq\emptyset$ and $A_{I\!I}\neq\emptyset$, then $A_{I}\cap A_{I\!I}\neq\emptyset$; say
$(x',y')\in A_{I}\cap A_{I\!I}$.  Let $(x,y)\in A_{I}$.  Then
$\tilde Q(t|x,y) = \tilde Q(t|x',y) = \tilde Q(t|x',y')$ for all~$t$.  Similarly, if
$(x,y)\in A_{I\!I}$.  Then $\tilde Q(t|x,y) = \tilde Q(t|x,y') = \tilde Q(t|x',y')$ for all~$t$.
Thus, all conditional distributions of $t$ given any $(x,y)\in\Xcal\times\Ycal$ are identical, and
so $A_{I} = A_{I\!I}=\Xcal'\times\Ycal'$.

The theorem now follows from the following observation: if $A_{I} = \Xcal'\times\Ycal'$, then \(\ind[\tilde{Q}]{T}{X}[Y]\), and if $A_{I\!I}=\Xcal'\times\Ycal'$, then \(\ind[\tilde{Q}]{T}{Y}[X]\).
\end{proof}

As a corollary to Theorem~\ref{thm:non-uniqueness-cardinalities-CI}:
\begin{thm}\label{thm:inner-uniqueness}
  Let \(\tilde{Q} \in \intDeltap \cap \arg\max_{Q\in \Delta_P}H(T|X,Y) \), and assume that $\Delta_{P}$ has full support. Then  \(\tilde{Q}\) is not unique if
  \begin{enumerate}
    \item \(\ind[\tilde{Q}]{T}{X}[Y]\) and \(|\Xcal| \geq 3\) or
    \item \(\ind[\tilde{Q}]{T}{Y}[X]\) and \(|\Ycal| \geq 3\).
  \end{enumerate}
  Equivalently, \(\tilde{Q}\) is not unique
  \begin{itemize}
  \item when $UI(T:X\setminus Y) = 0$ and $|\Ycal|>2$, or
  \item when $UI(T:Y\setminus X) = 0$ and $|\Xcal|>2$.
  \end{itemize}
\end{thm}

\subsection{The case of restricted support}
\label{sec:binary-T-restricted-support}

With a little more effort, the analysis of Theorem~\ref{thm:cond_indep_binary_T} extends to the case
where $\Delta_{P}$ has restricted support.  For any $t\in\Tcal'=\{0,1\}$ let
$\Xcal_{t}=\supp(P(X|T=t))$ and $\Ycal_{t}=\supp(P(Y|T=t))$.  Lemma~\ref{lem:suppDeltaPt} says that
$\supp(\Delta_{P,t}) = \Xcal_{t}\times\Ycal_{t}$.

For any $t\in\Tcal$ let $\bar t=1-t$.
% Let $(t,\bar t)$ be a permutation of $\Tcal'$.
If $x\notin\Xcal_{t}$, then $P(T=\bar t|X=x) = 1$.
Therefore, $\ind TY[\{X=x\}]$ for all $x\in\Xcal\setminus\Xcal_{t}$.  Similarly, $\ind TX[\{Y=y\}]$
for all $y\in\Ycal\setminus\Ycal_{t}$.  Thus, to prove that $\ind TY[X]$, say, it suffices to look
at $\Xcal_{0}\cap\Xcal_{1}$.

\begin{lemma}
  \label{lem:binary-T-XtYt}
  \begin{enumerate}
  \item If $\Xcal_{0}\cap\Xcal_{1}=\emptyset$, then $\ind[Q]TY[X]$ for any~$Q\in\Delta_{P}$.
  \item If $\Ycal_{0}\cap\Ycal_{1}=\emptyset$, then $\ind[Q]TX[Y]$ for any~$Q\in\Delta_{P}$.
  \item Suppose that $\Xcal_{0}\cap\Xcal_{1}\neq\emptyset\neq\Ycal_{0}\cap\Ycal_{1}$.
    \begin{enumerate}
    \item If there exists $t\in\Tcal'$ that satisfies $\Xcal_{t}\setminus\Xcal_{\bar t}\neq\emptyset$ and
      $\Ycal_{t}\setminus\Ycal_{\bar t}\neq\emptyset$, then $\arg\max_{Q\in\Delta_{P}}H(T|X,Y)$ does not intersect the interior $\intDeltap$.
    \item If $\Xcal_{t}\setminus\Xcal_{\bar t}\neq\emptyset$ and if there exists
      $Q^{*}\in\intDeltap\cap\arg\max_{Q\in\Delta_{P}}H(T|X,Y)$, then
      $\ind[Q^{*}]TY[\{X,Y\in\Ycal_{t}\}]$ (i.e., with respect to $Q^{*}$, $T$ is independent of $Y$
      given $X$, given that $Y\in\Ycal_{t}$).
    \item If $\Ycal_{t}\setminus\Ycal_{\bar t}\neq\emptyset$ and if there exists
      $Q^{*}\in\intDeltap\cap\arg\max_{Q\in\Delta_{P}}H(T|X,Y)$, then
      $\ind[Q^{*}]TX[\{Y,X\in\Xcal_{t}\}]$.
    \end{enumerate}
  \end{enumerate}
\end{lemma}
\begin{proof}
  Statements (1) and (2): If $\Xcal_{0}\cap\Xcal_{1}=\emptyset$, then $T$ is a function of $X$ for
  any~$Q\in\Delta_{P}$, whence $\ind[Q]TY[X]$. Statement~(2) follows similarly.
  
  Statement (3a): Let $x_{0}\in\Xcal_{0}\cap\Xcal_{1}$, $y_{0}\in\Ycal_{0}\cap\Ycal_{1}$,
  $x_{1}\in\Xcal_{t}\setminus\Xcal_{\bar t}\neq\emptyset$ and
  $y_{1}\in\Ycal_{t}\setminus\Ycal_{\bar t}\neq\emptyset$.  Suppose that $q\in\intDeltap$.
  Then $Q(t,x_{0},y_{0})>0$ and $Q(\bar t,x_{0},y_{0})>0$, whence $Q(t|x_{0},y_{0})\neq 1$.
  Then the derivative of $H(T|X,Y)$ in the direction of $\gamma_{t;x_{0},x_{1};y_{0},y_{1}}$ is
  \begin{equation*}
    \log\frac{Q(t|x_{0},y_{0})Q(t|x_{1},y_{1})}{Q(t|x_{0},y_{1})Q(t|x_{1},y_{0})}
    = \log Q(t|x_{0},y_{0})\neq 0.
  \end{equation*}

  Statement (3b): % By 3(a), $\Ycal_{t}\setminus\Ycal_{\bar t}=\emptyset$.
  If $|\Ycal_{t}|=1$, then $Y$ is constant when conditioning on~$Y\in\Ycal_{t}$, whence the
  conclusion holds trivially.  Let $y_{0},y_{1}\in\Ycal_{t}$ with $y_{0}\neq y_{1}$, let
  $x_{0}\in\Xcal_{0}\cap\Xcal_{1}$, and let $x_{1}\in\Xcal_{t}\setminus\Xcal_{\bar t}\neq\emptyset$.
  The derivative of $H(T|X,Y)$ at $Q^{*}$ in the direction of $\gamma_{t;x_{0},x_{1};y_{0},y_{1}}$
  is
  \begin{equation*}
    \log\frac{Q^*(t|x_{0},y_{0})Q^*(t|x_{1},y_{1})}{Q^*(t|x_{0},y_{1})Q^*(t|x_{1},y_{0})}
    = \log\frac{Q^*(t|x_{0},y_{0})}{Q^*(t|x_{0},y_{1})}.
  \end{equation*}
  By assumption, this derivative vanishes at $Q^{*}$, whence
  $Q^*(t|x_{0},y_{0})=Q^*(t|x_{0},y_{1})$, which proves the statement.
\end{proof}

\begin{thm}
  \label{thm:cond_indep_binary_T_supp}
  Let $T$ be binary, and suppose that $Q^{*}\in\arg\max_{Q\in\Delta_{P}}H(T|X,Y)$ lies in
  $\intDeltap$.
  \begin{itemize}
  \item If $\Xcal_{0}=\Xcal_{1}$ and $\Ycal_{0}\neq\Ycal_{1}$, then $\ind[Q^{*}]TX[Y]$.
  \item If $\Ycal_{0}=\Ycal_{1}$ and $\Xcal_{0}\neq\Xcal_{1}$, then $\ind[Q^{*}]TY[X]$.
  \end{itemize}
\end{thm}
\begin{proof}
  The theorem follows from Lemma~\ref{lem:binary-T-XtYt}.
\end{proof}

\subsection{Statistics for uniqueness and support of optimizers for binary T}
\label{sec:binary-T-uniqueness-support}

To better understand whether the optimizer typically lies in the interior of \(\Delta_P\) and
whether it is typically unique, we uniformly sampled joint distributions
$P\in \Delta_{\Tcal,\Xcal,\Ycal}$ for binary $\Tcal$ and different cardinalities of \(|\Xcal|,|\Ycal|\).
% In this section, we report statistics how often the optimum is found in the interior of \(\Delta_P\) and if it is unique in dependence of the cardinalities \(|\Xcal|,|\Ycal|\), for uniformely sampled distributions $P\in \Delta_{\Tcal,\Xcal,\Ycal}$.
Uniform sampling from \(\Delta_{T,X,Y} \) was performed with Kraemers' method \citep{smith2004sampling}.
Based on 10000 samples, the following percentage of optima were found in the interior of \(\Delta_P\):
\begin{center}
\begin{tabular}{ccccc}
  \(|\Xcal|/|\Ycal|\) & 2 & 3 & 4 & 5 \\
  \midrule
  2 & 77.6 &49.3 & 76.3 & 81.4\\
  3 & - & 52.7 & 58.4 & 63.8\\
  4 & - & - & 57.0 &56.3 \\
  5 & - & - & - & 53.1
\end{tabular}
\end{center}

The percentage of solutions found in the interior of \(\Delta_P \) decreases with increasing cardinality of \(|\Xcal|\) and \(|\Ycal|\).  The following table lists the percentages for \(|\Xcal|=|\Ycal|=k\) over 1000 samples for different values of~$k$.
\begin{center}
  \begin{tabular}{ccccccccc}
    $k$:& 6&8&10&12&14&16&18&20 \\ \midrule
    optimizer in interior [\%]:& 47.8&43.9&41.0&37.4&37.3&37.5&32.5&29.1
  \end{tabular}
\end{center}

Under uniform sampling, all sampled distributions have full support. In accordance with Theorem~\ref{thm:inner-uniqueness}, we do not find unique optima in the interior of \(\Delta_P\), except when the cardinalities are \(2\times 2\times k\).  In the $2\times2\times k$-case, the percentage of samples where we found unique optimizers are (10,000 samples per $k$):
\begin{center}
  \begin{tabular}{rcccccc}
    $k$: &2&3&4&5&6&10\\ \midrule
    optimizer unique [\%]: & 100&31.2&7.4&2.4&0.1&0
\end{tabular}
\end{center}

\subsection{Visualization of the 2x2x3 case}

For the all binary case, the geometry of optimization domain $\Delta_P$ is generically a rectangle and can readily be visualized, see \citet{bertschinger2014quantifying}. In this section we aim to illustrate the features of the optimization domain for the next larger case $|\Tcal|=2,|\Xcal|=2,|\Ycal|=3$. In this case, the four-dimensional optimization domain \(\Delta_P = \Delta_{P,0} \times \Delta_{P,1}\) is the direct product of two two-dimensional polytopes. We parameterize elements \(Q\in \Delta_P \) by \(Q = Q_0 +
P_T(0)(g_{0,0,0}\gamma_{0;0,1,0,2}+g_{0,0,1}\gamma_{0;0,1,1,2})+P_T(1)(g_{1,0,0}\gamma_{1;0,1,0,2}+g_{1,0,1}\gamma_{1;0,1,1,2})\).
Figure~\ref{fig:vis-223}  visualizes \(\Delta_{P,0},\Delta_{P,1}\), and the projections of $\arg\max_{Q\in\Delta_{P}} H_{Q}(T|XY)$ for three different distributions sampled from the unit simplex. In (a) and (b), $\arg\max_{Q\in\Delta_{P}} H_{Q}(T|XY)$ is a singleton in the interior or on the boundary of \(\Delta_P\). Note that in case (b) both projections of the optimizer lie at the boundary of \(\Delta_{P,0},\Delta_{P,1}\), in agreement with Lemma~\ref{lem:maxi-boundary}
In (c), there exists no unique optimizer, but conditional independence \(\ind[Q^*]{T}{X}[Y]\) holds for all \(Q^*\) on the line segment between the boundary points in \(\Delta_{P,0} \) and \(\Delta_{P,1}\) and every such \(Q^*\in \arg\max_{Q\in\Delta_{P}} H_{Q}(T|XY)\).

\begin{figure*}
  \centering
    \begin{tabular}{c}
      (a) Unique optimizer in interior of $\Delta_P$\\
      \includegraphics{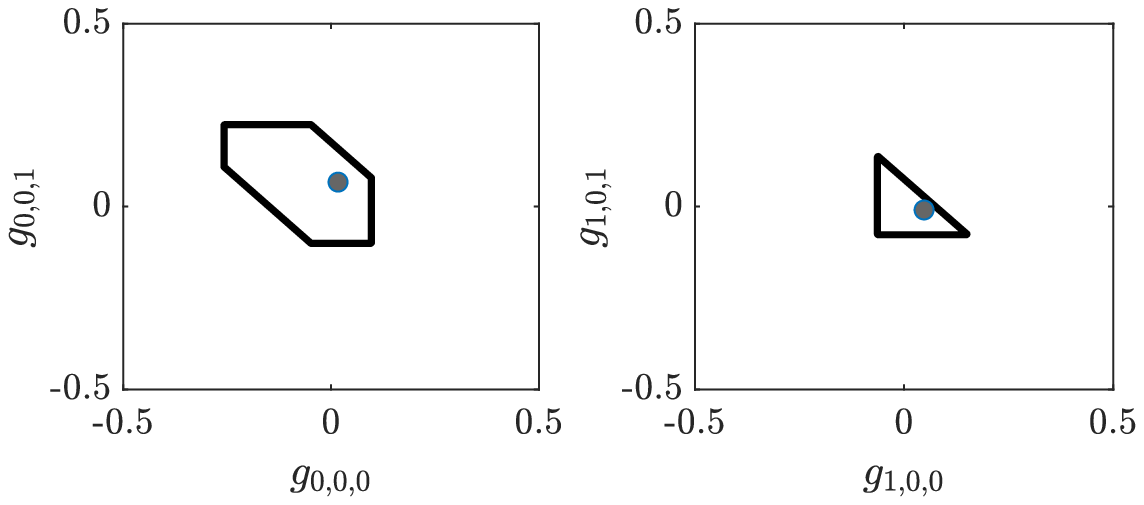}\\
      (b) Unique optimizer at boundary of $\Delta_P$\\ 
      \includegraphics{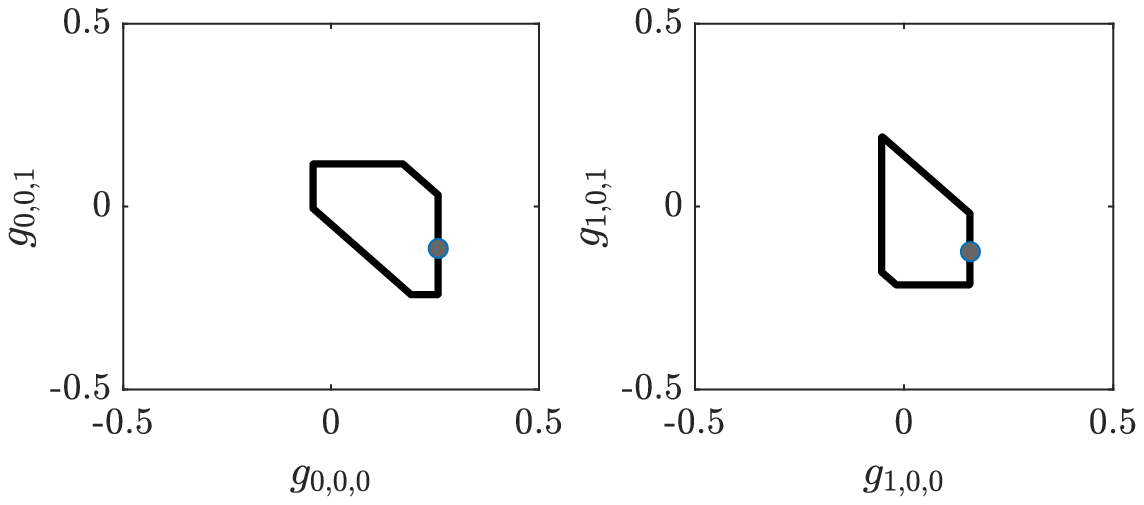}\\
      (c) \(\arg \max_{Q\in\Delta_{P}} H_{Q}(T|XY)\) given by a line segment\\
      \includegraphics{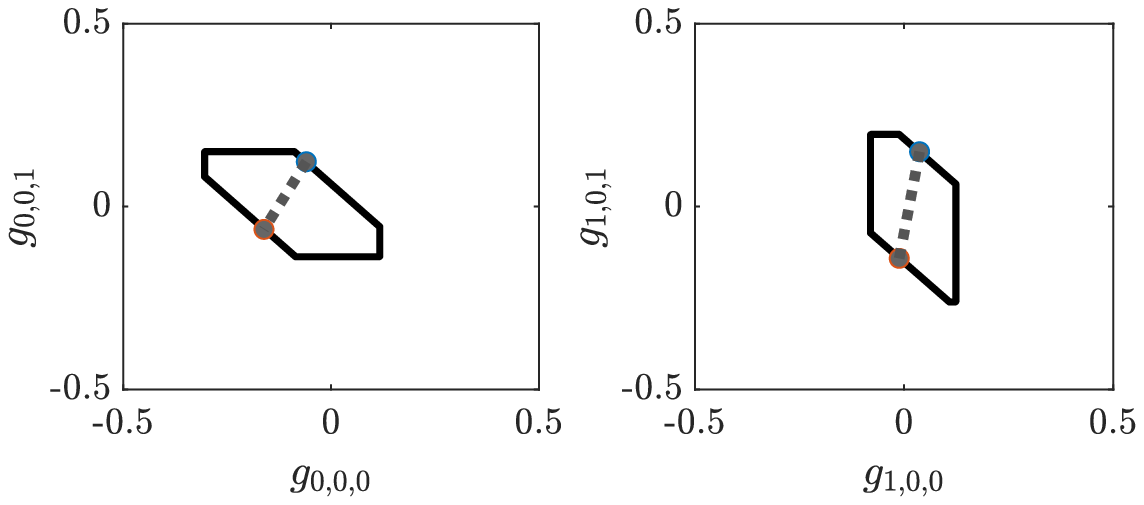}\\
    \end{tabular}
    \caption{\label{fig:vis-223}Three examples of \(\Delta_P\) and $\arg\max_{Q\in\Delta_{P}} H_{Q}(T|XY)$ for
      \(|\Tcal|=2\),\(|\Xcal|=2\),\(|\Ycal|=3\). Left plots show \(\Delta_{P,0} \), while \(\Delta_{P,1}\) is shown on the right
      side. Boundaries of \(\Delta_P\) are marked by black lines, $\arg\max_{Q\in\Delta_{P}} H_{Q}(T|XY)$ is marked in
      gray. (a) There exists a unique optimizer in the interior of \(\Delta_P\). (b) The unique optimizer lies at the boundary of
      \(\Delta_P\). Note that both projections of the optimizer lies at the boundary of \(\Delta_{P,0},\Delta_{P,1}\). (c)
      Gray lines mark the projections of $\arg\max_{Q\in\Delta_{P}} H_{Q}(T|XY)$ to \(\Delta_{P,0}\) and \(\Delta_{P,1}\).}
    \end{figure*}

\section{The all binary case}
\label{sec:binary-case}

If $X$, $Y$ and $T$ are all binary, \(\Delta_{\Tcal,\Xcal,\Ycal}\) has 7 dimensions, which split in
5 dimensions for \(V\) and 2 dimensions for \(\Delta_{P}\).
In this case it is possible to explicitly describe $\arg\max_{Q\in\Delta_{P}}H_{Q}(T|X,Y)$.  This description will be developped throughout this chapter and summarized at the end of this section in Theorem~\ref{thm:all-binary-cases}.

Throughought this section we assume that $\Tcal'=\{0,1\} = \Xcal = \Ycal$. 
In the following, \(V\) is parameterized by the variables
\begin{align}
  \label{eq:param-2-2-2}
  a = P_T(0), &&
  \begin{aligned}
    b = P_{X|T}(0|0), \\
    c = P_{X|T}(0|1),
  \end{aligned}
  &&
  \begin{aligned}
    d = P_{Y|T}(0|0),  \\ 
    e = P_{Y|T}(0|1),
  \end{aligned}
\end{align}
and by the coefficients \(g_1,g_2\) of \(a \gamma_{0;0,1;0,1},(1-a) \gamma_{1;0,1;0,1}\).
Table~\ref{tab:parameterizations}
makes the parametrization~\eqref{eq:decomposition} explicit.
\begin{table}% [htbp]
 \centering
\begin{tabular}{rrrl}
  \(T\) & \(X\) & \(Y\) & \(P(t,x,y)\)\\
  \midrule
0 & 0 & 0 & \(a(bd + g_1)\)\\
0 & 0 & 1 & \(a(b(1-d) - g_1)\)\\
0 & 1 & 0 & \(a((1-b)d - g_1)\)\\
0 & 1 & 1 & \(a((1-b)(1-d) + g_1)\)\\
1 & 0 & 0 & \((1-a)(ce + g_2)\)\\
1 & 0 & 1 & \((1-a)(c(1-e) - g_2)\)\\
1 & 1 & 0 & \((1-a)((1-c)e - g_2)\)\\
1 & 1 & 1 & \((1-a)((1-c)(1-e) + g_2)\)\\
\end{tabular}
 \caption{Parameterization of \(2\times 2 \times 2\) distributions \label{tab:parameterizations}}
\end{table}

\(\Delta_{P}\) is a rectangle.  The allowed parameter domain is
\begin{align*}
  - %a
  \min{\{bd,(1-b)(1-d)\}} &\leq g_1 \leq %a
                            \min{\{b(1-d),(1-b)d\}}\\
  - %(1-a)
  \min{\{ce, (1-c)(1-e)\}} &\leq g_2 \leq %(1-a)
                             \min{\{c(1-e),(1-c)e\}} \text{ .}
\end{align*}
The lower and upper bounds on \(g_i\) will be denoted by
\(g_{i_{\min}}\) and \(g_{i_{\max}}\) respectively.

The following holds:
\begin{enumerate}
\item $\Delta_{p,0}$ is a singleton iff $b\in\{0,1\}$ or $d \in\{0,1\}$.
\item $\Delta_{p,1}$ is a singleton iff $c\in\{0,1\}$ or $e \in\{0,1\}$.
\item $\Delta_{P}$ is a singleton iff both conditions are met.
  Thus, $\Delta_{P}$ degenerates to a single point precisely in the following four cases:
  \begin{enumerate}
  \item $H(X|T)=0$;
  \item $H(Y|T)=0$;
  \item $H(X|T=0)=0$ and $H(Y|T=1)=0$;
  \item $H(X|T=1)=0$ and $H(Y|T=0)=0$.
  \end{enumerate}
\end{enumerate}

In the all-binary case, Theorem~\ref{thm:cond_indep_binary_T} slightly generalizes:

\begin{thm}
  \label{thm:cond_indep_all_binary}
  Let $X,Y,T$ be binary.
  Suppose that ${\Delta}_P$ is not a singleton in case (c) or (d).
  If \(\tilde{Q} = \arg\max_{Q\in \Delta_{P}}{H_Q(T|X,Y)} \in \intDeltap\), then
  \(\ind[\tilde{Q}]{T}{X}[Y]\) or \(\ind[\tilde{Q}]{T}{Y}[X]\).
\end{thm}
\begin{remark}
  Example~\ref{ex:binary-singleton} shows that the conclusion does not in general  hold in the singleton cases (c) and (d).
\end{remark}
\begin{proof}
  The singleton cases (a) and (b) are trivial, and the remaining cases follow from Theorem~\ref{thm:cond_indep_binary_T_supp}.
\end{proof}

In the all-binary case, uniqueness can be completely characterized:

\begin{thm}
  $\arg\max_{Q\in \Delta_{P}}{H_Q(T|X,Y)}$ is unique, unless $b=c$ and $d=e$.
\end{thm}

\begin{proof}

\comment{\color{green}
The equations in case $I$ are equivalent to
\begin{equation*}
  \frac{\tilde Q(1 x y)}{\tilde Q(0 x y)}
  = \frac{\tilde Q(1 x' y)}{\tilde Q(0 x' y)}
  \qquad\Longleftrightarrow\qquad
  \tilde Q(1 x y) \tilde Q(0 x' y)
  = \tilde Q(1 x' y) \tilde Q(0 x y)
\end{equation*}
for all $x,x'\in\Xcal$, $y\in\Ycal$.
Insert $\tilde Q(txy) = p(t)\big(p(x|t)p(y|t) + g(txy)\big)$ and divide by $p_{T}(1)p_{T}(0)$:
\begin{multline*}
  p(x|1)p(x'|0)p(y|1)p(y|0)
  + g(1xy) p(x'|0)p(y|0)
  + g(0x'y) p(x|1)p(y|1)
  + g(1xy) g(0x'y)
  \\
  =
  p(x'|1)p(x|0)p(y|1)p(y|0)
  + g(1x'y) p(x|0)p(y|0)
  + g(0xy) p(x'|1)p(y|1)
  + g(1x'y) g(0xy)
\end{multline*}

We may suppose that $\ind[p]TX[Y]$, whence $p(txy) = p_{T}(t)p(y|t)p(x|y)$.  Insert
$\tilde Q(txy) = p_{T}(t)p(y|t)p(x|y) + p_{T}(t)g(txy)$ and divide by $p_{T}(1)p_{T}(0)$:
\begin{multline*}
  p(y|1)p(y|0)p(x|y)p(x'|y)
  + g(1xy) p(y|0)p(x'|y)
  + g(0x'y) p(y|1)p(x'|y)
  + g(1xy) g(0x'y)
  \\
  =
  p(y|1)p(y|0)p(x|y)p(x'|y)
  + g(1x'y) p(y|0)p(x|y)
  + g(0xy) p(y|1)p(x'|y)
  + g(1x'y) g(0xy),
\end{multline*}
which simplifies to
\begin{multline*}
  g(1xy) p(y|0)p(x'|y)
  + g(0x'y) p(y|1)p(x'|y)
  + g(1xy) g(0x'y)
  \\
  =
  g(1x'y) p(y|0)p(x|y)
  + g(0xy) p(y|1)p(x'|y)
  + g(1x'y) g(0xy).
\end{multline*}
}

If $\arg\max_{Q\in \Delta_{P}}{H_Q(T|X,Y)}$ is not unique, then $\arg\max_{Q\in\intDeltap}{H_Q(T|X,Y)}$ is not unique either (by Lemma~\ref{lem:maxi-boundary}), so we may restrict attention to maximizers in the interior of~$\Delta_{P}$.
Thus, we assume that $b,c,d,e\notin\{0,1\}$.

First assume that $\Delta_{P}$ has full support.
As shown in Theorem~\ref{thm:cond_indep_binary_T} and its proof, there are two cases $I$ and $I\!I$ to consider.
Inserting the parameterization from above 
and using the injectivity of \(\frac{1}{1+x}\) leads for case \(I\) to the
equations \footnote{No solutions exist for which one denominator equals 0. The same applies for case \(I\!I\).}
\begin{align*}
\frac{ce + g_2}{bd + g_1} &= \frac{(1-c)e - g_2}{(1-b)d - g_1} \\
\frac{c(1-e) - g_2}{b(1-d) - g_1} &= \frac{(1-c)(1-e) + g_2}{(1-b)(1-d) + g_{1}} \text{ ,}
\end{align*}
which simplify to 
\begin{align*}
g_2 d - g_1 e &= de(b - c) \\
g_1 (1-e) - g_2 (1-d) & =  (1-d)(1-e)(b - c) \text{ .}
\end{align*}
Rearranging for \(g_1,g_2\) leads to
\begin{equation} \label{eq:case1}
  \begin{aligned}
    g_1 (d - e) &= d(b - c)(1 - d) \\
    g_2 (d - e) &= e(b-c)(1-e) \text{ .}
  \end{aligned}
\end{equation}
For  \(d \neq e\), there exists a unique solution, and for \(b = c\), the
optimum is \(Q_0\) itself.
For $d=e$, there only exists a solution if $b=c$.

Similarly, case \(I\!I\)  reduces to 
\begin{align*}
g_2 b - g_1 c &= bc(d-e) \\
g_1 (1-c) - g_2 (1-b) &= (1-b)(1-c)(d-e)
\end{align*}
and rearranging for \(g_1,g_2\) gives
\begin{equation} \label{eq:case2}
  \begin{aligned}
    g_1 (b-c) &= b(d - e)(1 - b) \\
    g_2 (b-c) &= c(d - e)(1-c) \text{ .}
  \end{aligned}
\end{equation}
Again, there exists a unique solution for \(b \neq c\) and \(Q_0\)
is the optimum for \(d = e\).

Now assume that $\Delta_{P}$ is a line.  Following the proof of
Theorem~\ref{thm:cond_indep_all_binary}, assume that $b=0$.
Plugging the parametrization from above 
into the equality $Q(1|10) = Q(1|11)$
gives
\begin{equation*}
  \frac{(1 - a) \big((1 - c) e - g_{2}\big)}{(1 - a) \big((1 - c) e - g_{2}\big) + P(010)}
  =
  \frac{(1 - a) \big((1 - c) (1 - e) + g_2\big)}{(1 - a) \big((1 - c) (1 - e) + g_2\big) + P(011)}.
\end{equation*}
If $P(010) = 0$, then $P(011) = 0$, and conversely; otherwise, this equation has no solution.  In this case $P(010) = P(011) = 0$, the sum $P(01) = P(010) + P(011) = a$ vanishes, which contradicts $\Tcal'=\{0,1\}$.  Thus, $P(010) \neq 0$ and $P(011) \neq 0$.
Using injectivity of $x\mapsto \frac{1}{1 + x}$ and cancelling $(1-a)$, this is equivalent to
\begin{equation}
  \label{eq:case3}
  \frac{(1 - c) e - g_{2}}{P(010)}
  =
  \frac{(1 - c) (1 - e) + g_2}{P(011)}.  
\end{equation}
This equation is linear in $g_{2}$ and has a single unique solution, since the coefficient $\frac{1}{P(010)} + \frac{1}{P(011)}$ in front of $g_{2}$ is positive.
\end{proof}

Only the case where the maximizer lies on the boundary of
\(\Delta_{P}\) remains to be analyzed.

\begin{thm}
  Assume that  \(\tilde{Q} =
   \arg\max_{Q\in \Delta_{P}}{H_Q(T|X,Y)}\) lies at
the boundary of~\(\Delta_{P}\). Then, it is attained either at
\((g_{1_{\min}},g_{2_{\min}})\) or \((g_{1_{\max}},g_{2_{\max}})\).
\end{thm}

\begin{proof}
  If $\Delta_{P}$ is degenerate, then either $g_{1_{\min}}=g_{1_{\max}}$ or $g_{2_{\min}}=g_{2_{\max}}$, and the theorem becomes trivial.
  Otherwise, the statement follows from Lemma~\ref{lem:maxi-boundary}.
\end{proof}

The following theorem sums up the different possibilities.
\begin{thm}
  \label{thm:all-binary-cases}
  For non-constant binary random variables $X$, $Y$ and $T$, there are five cases:
\begin{enumerate}
\item $b=c$ and $d=e$.  In this case, $\ind XY[T]$
  and $\arg\max_{Q\in \Delta_{P}}H_Q(T|X,Y)$ is not unique, but consists of the diagonal of
  $\Delta_{P}$.
\item $\ind[\tilde{Q}]{T}{X}[Y]$ for the unique $\tilde Q = \arg\max_{Q\in \Delta_{P}}H_Q(T|X,Y)$. % (which can be computed by solving~\eqref{eq:case1}).
\item \(\ind[\tilde{Q}]{T}{Y}[X]\) for the unique $\tilde Q = \arg\max_{Q\in \Delta_{P}}H_Q(T|X,Y)$. % (which can be computed by solving~\eqref{eq:case2}).
\item The unique maximizer lies at $(g_{1_{\min}},g_{2_{\min}})$.
\item The unique maximizer lies at $(g_{1_{\max}},g_{2_{\max}})$.
\end{enumerate}
\end{thm}
\begin{remark}
  (1) The last four cases in Theorem~\ref{thm:all-binary-cases} intersect.  For example, the intersection of the last four cases contains the
distribution $\frac12 \delta_{000} + \frac12 \delta_{111}$ (see~\citep{Fink11:Binomial_ideal_of_intersection_axiom,RauhAy14:Robustness_and_systems_design} for a discussion of the intersection of cases (2) and (3)).

(2) In cases 2.\ and 3., if $0<b,c,d,e <1$, then $\tilde Q$ can be computed by
solving~\eqref{eq:case1} or~\eqref{eq:case2}.  If $b=0$, then $\tilde Q$ can be computed in cases 2.\ and 3.\ by solving~\eqref{eq:case3}.  Similar equations can be obtained if $b=1$ or if any of $c,d,e$ lies in $\{0,1\}$.

(3) The five cases can be distinguished by polynomial inequalities among the parameters
$a,b,c,d,e$.  Therefore, the five cases correspond to five semi-algebraic sets of probability
distributions.
For example, case (2) holds if and only if the unique solution
$(g_{1},g_{2})$ to~\eqref{eq:case1} satisfies
$g_{i_{\min}}\le g_{i}\le g_{i_{\max}}$ for $i=1,2$, which can be formulated as eight
polynomial inequalities.
\end{remark}

% These results make it possible to exactly solve

% \(\arg\max_{Q \in \Delta_{P}}H_Q(T|X,Y)\) by
% checking whether the solutions of \eqref{eq:case1}, \eqref{eq:case2} or~\eqref{eq:case3}
% lie in \(\Delta_{P}\) and otherwise using the maximum of
% \(H(T|X,Y)\) at
% \((g_{1_{\min}},g_{2_{\min}})\) and
% \((g_{1_{\max}},g_{2_{\max}})\).

\section{Examples}
\label{sec:examples}

\begin{example}[For ternary $T$, maximizers with full support need not satisfy CI statements]
  \label{ex:ternary-no-CI}
Let $X,Y$ be binary random variables with $P(X, Y)$ arbitrary (of full support), and let $T$ be ternary with
\begin{align*}
  (P(T=1|X=x,Y=y))_{x,y} &=
  \begin{pmatrix}
    \frac13 & \frac12 \\
    \frac1{12} & \frac18
  \end{pmatrix},
  \\
  (P(T=2|X=x,Y=y))_{x,y} &= 
  \begin{pmatrix}
    \frac13 & \frac18 \\
    \frac5{24} & \frac5{64}
  \end{pmatrix},
  \\
  (P(T=3|X=x,Y=y))_{x,y} &=
  \begin{pmatrix}
    \frac13 & \frac38 \\
    \frac{17}{24} & \frac{51}{64}
  \end{pmatrix}
\end{align*}

Then $P$ minimizes $I_{Q}(T:X|Y)$ on $\Delta_{P}$ (cf. Remark~\ref{rem:sum-of-rk-one}), and one can
check that $P$ is the unique minimizer on~$\Delta_{P}$ (it is impossible to find a line through $P$ in
$\Delta_{P}$ such that the two points at which this line hits the boundary satisfy the conclusion of
Lemma~\ref{lem:maxi-boundary}). $P$ has full support, but there is no conditional independence
statement.
\end{example}

\begin{example}[An illustration of Theorem~\ref{thm:non-uniqueness-v2}]
  \label{ex:blockwise-non-uniqueness}
  Consider the distributions
  \begin{center}
    \begin{tabular}{cccc}
      $x$ & $y$ & $t$ & $P(x,y,t)$ \\
      \midrule
      0 & 0 & 0 & $\tfrac16$ \\
      0 & 0 & 1 & $\tfrac16$ \\
      1 & 1 & 0 & $\tfrac16$ \\
      1 & 1 & 1 & $\tfrac16$ \\
      2 & 2 & 0 & $\tfrac19$ \\
      2 & 2 & 1 & $\tfrac29$
    \end{tabular}
    \hfil
    \begin{tabular}{cccc}
      $x$ & $y$ & $t$ & $P'(x,y,t)$ \\
      \midrule
      0 & 1 & 0 & $\tfrac16$ \\
      0 & 1 & 1 & $\tfrac16$ \\
      1 & 0 & 0 & $\tfrac16$ \\
      1 & 0 & 1 & $\tfrac16$ \\
      2 & 2 & 0 & $\tfrac19$ \\
      2 & 2 & 1 & $\tfrac29$
    \end{tabular}
  \end{center}
  Then $\ind[P]{T}{Y}[X]$ and $\ind[P']{T}{Y}[X]$ as well as $\ind[P]{T}{X}[Y]$ and $\ind[P']{T}{X}[Y]$, and $P'\in\Delta_{P}$.  It follows that
  $I_{P}(T:Y|X) = I_{P'}(T:Y|X) = 0$, whence $P$ and $P'$ are both minimizers.  The same holds true for any convex
  combination of $P$ and~$P'$.  Note that $P$ and $P'$ (more generally: any convex combination of $P$ and $P'$) have
  restricted support: the probability of $\big\{X=2, Y\neq 2\big\}$ vanishes.
  On the other hand, $\supp(\Delta_{P})$ is full.
\end{example}

\begin{example}[The all-binary case where $\Delta_{P}$ is a line]
  \label{ex:suppDp-not-full}
  Consider the $2\times2\times2$ distribution given by $e=0$ and $a,b,c,d = \tfrac12$
  \begin{center}
\begin{tabular}{cccc}
  \(T\) & \(X\) & \(Y\) & \(P(t,x,y)\)\\
  \midrule
0 & 0 & 0 & $\tfrac18$ \\
0 & 0 & 1 & $\tfrac18$\\
0 & 1 & 0 & $\tfrac18$\\
0 & 1 & 1 & $\tfrac18$\\
1 & 0 & 1 & $\tfrac14$\\
1 & 1 & 1 & $\tfrac14$
\end{tabular}
  \end{center}
  $\Delta_{P}$ degenerates to a line $P + g_{1}\gamma_{0;0,1;0,1}$ with support $-\tfrac18\leq g_1 \leq \tfrac18$.
  The conditional entropy is
  \begin{align*}
    H_{g_{1}}(T|X,Y)
    & = (\tfrac38 - g_{1}) H_{g_{1}}(T|0,1) + (\tfrac38 + g_{1}) H_{g_{1}}(T|1,1) \\
    & = (\tfrac38 - g_{1}) h\bigg(\frac{\tfrac18 - g_{1}}{\tfrac38 - g_{1}},\frac{\tfrac14}{\tfrac38 - g_{1}}\bigg) + (\tfrac38 + g_{1}) h\bigg(\frac{\tfrac18 + g_{1}}{\tfrac38 + g_{1}},\frac{\tfrac14}{\tfrac38 + g_{1}}\bigg).
  \end{align*}
  By symmetry and Lemma~\ref{lem:unique-conditional}, the unique maximizer of $H_{g_{1}}(T|X,Y)$ lies at $g_{1}=0$, that is, $P$
  is the unique solution to the optimization problem.  In this case, $P$ equals~$Q_{0}$; that is, $\ind[P]XY[T]$ holds.  Moreover, $\ind[P]TX[Y]$ holds.
\end{example}

\begin{example}[The all-binary case where $\Delta_{P}$ is a singleton]
  \label{ex:binary-singleton}
  Consider the $2\times2\times2$ distribution given by $b=e=1$ and $a,c,d = \tfrac12$:
  \begin{center}
    \begin{tabular}{cccc}
      \(T\) & \(X\) & \(Y\) & \(P(t,x,y)\)\\
      \midrule
      0 & 0 & 0 & $\tfrac14$ \\
      0 & 0 & 1 & $\tfrac14$ \\
      1 & 0 & 0 & $\tfrac14$ \\
      1 & 1 & 0 & $\tfrac14$
    \end{tabular}
  \end{center}
  Here, $\Delta_{P}$ is a singleton.
  Neither $\ind[P]TY[X]$ nor $\ind[P]TX[Y]$ holds.
\end{example}

\section{Conclusions}
\label{sec:conclusions}

In this work we investigated uniqueness and support of the solutions
to the optimization problem underlying the definition of the unique
information function $UI(T:X\setminus Y)$ defined by
\citet{bertschinger2014quantifying}. This optimization problem
consists of maximizing the conditional entropy \(H(T|XY)\) over the
space of probability distributions with fixed pairwise \(TX,TY\)
marginals. We showed that this conditional entropy is not strictly
concave in exactly the directions in which \(P(T|XY)\) is constant.
From this we showed that all optima that are attained in the interior
of the optimization space which have full support are not unique if
\(|\Tcal|<\max(|\Xcal|,|\Ycal|)\) and identified sufficient conditions
for non-uniqueness that relate to independence statements. If the
variable \(\Tcal\) is binary we showed partial converses of these
results. In this case, vanishing of the directional derivatives of the
\(H(T|XY)\) implies a conditional independence $\ind TY[X]$ or
$\ind TX[Y]$ and thus vanishing of the corresponding unique
informations. Imposing such an independence relation on the
optimization domain led to a set of linear constraints. Thus, by
solving this linear problems we solve the optimization problem if
there exists a solution in the interior, otherwise we reduce the
optimization domain to its boundary. Numerical experiments showed that
a noticeable fraction of distributions sampled uniformely from the
probability simplex have corresponding optima in the interior. This
fraction becomes smaller with growing cardinalities of
\(|\Xcal|,|\Ycal|\). We derived an analytical solution of the
optimization problem when all variables are binary. Whenever possible,
we gave extensions to the theorems relaxing the assumptions on the
support of the optima and gave examples showing that the assumptions
in our theorems are neccesary.

\section*{Authors' Contributions}

Work on this project was initiated by questions of JJ.  Initial results for the all binary case were obtained by MS and JJ.  MS and JR worked together to generalize and complete the results.  MS and JR wrote the paper.  All authors read and approved the final manuscript.

\section*{Acknowledgement}

\small
Maik Sch\"unemann received support from the SMARTSTART program and the DFG priority proram SPP 1665 (ER 324/3-1).
We thank Pradeep Kr.\ Banerjee, Eckehard Olbrich and Udo Ernst for helpful remarks.

\ifarxiv\else
\makesubmdate
\fi

\bibliographystyle{plainnat}
\bibliography{UIproperties}

\ifarxiv\else
\makecontacts
\fi

\end{document}